\documentclass[letterpaper,twocolumn,10pt]{article}
\usepackage{main}

\usepackage{tikz}
\usepackage{amsmath}
\usepackage{caption}
\usepackage{subcaption}
\usepackage{graphicx}
\usepackage{float}
\usepackage{tabularx}
\usepackage{mathptmx}  
\usepackage{amsthm}
\newtheorem{theorem}{Theorem}

\usepackage[linesnumbered,ruled,vlined]{algorithm2e}
\SetKwComment{Comment}{/* }{ */}

\usepackage{filecontents}

\begin{document}

\date{}

\title{\Huge \bf Routing for Large ML Models}

\author{
\rm Ofir Cohen$^1$, Jose Yallouz$^2$, Michael Schapira$^1$, Shahar Belkar$^3$, and Tal Mizrahi$^2$\\\rm $^1$Hebrew University of Jerusalem, $^2$Technion,\\ \rm $^3$Huawei Network.IO Innovation Lab
}

\maketitle

\begin{abstract}
Training large language models (LLMs), and other large machine learning models, involves repeated communication of large volumes of data across a data center network. The communication patterns induced by these training process exhibit high regularity and persistence, giving rise to significant opportunities for optimizing the manner in which flows are routed across the network. We present an algorithmic framework for \textit{quantifying} network-wide efficiency in the context of training LLMs (and other large-scale ML models), and for periodically \textit{optimizing} routing with respect to this global metric.
\end{abstract}

\section{Introduction}

Deep Neural Networks (DNNs), the workhorses of deep learning, have gained immense popularity across multiple application domains. More recently, massive scale DNNs have proven invaluable for generative AI applications that are growing in an unprecedented pace. OpenAI's ChatGPT, for instance, has reached $100$ million active users within three months of its release, making it the fastest-growing application ever~\cite{chatgpt-article}.
This growth has led to an urgent demand for efficient distributed DNN training systems.

Today’s DNN training platforms are typically built on top of traditional data center clusters, organized in traditional Clos network topologies~\cite{10.1145/1402958.1402967}.
To facilitate more communication-efficient DNN training, prior research has investigated how to reduce the size of the DNN parameters transmitted across the network~\cite{alistarh2017qsgd,8868152,chung2017accelerating,goyal2018accurate,hashemi2018tictac,iandola2016firecaffe,jayarajan2019prioritybased,jia2018highly,mudigere2023softwarehardware,harlappipedream,wang2019blink}, as well as parallelization strategies that take into account the available network bandwidth~\cite{addanki2019placeto,alistarh2017qsgd,jia2018data,harlappipedream,tarnawski2020efficient}. 
While these proposals co-optimize computation and communication, they do not consider the \emph{physical network topology} as an optimization dimension. In addition, these proposals consider a single training job, and not the common scenario where multiple jobs co-exist. Recent work, namely, TopoOpt~\cite{wang2022topoopt} considers the simultaneous optimization of switch-level topology and ML training parallelization strategy, but requires large changes to the network infrastructure (optical switching) and only has limited benefits when traffic from hosts under the same switch must reach hosts under multiple switches.


Our aim is to devise methodologies for the online adaptation of routing configurations in ML training clusters that improve global training efficiency and fairness. Our approach builds on two characteristics of ML training and modern networking:

\begin{itemize}

\item Traffic patterns induced by ML training tend to exhibit high regularity and predictability, with the same pairs of hosts/GPUs periodically exchanging similar volumes of traffic.

\item The SmartNIC at a host can determine which path traffic will traverse. This can be implemented either using segment routing (e.g., SRv6), or by controlling the manipulating the packet header fields to ensure that ECMP hashing at switches maps the packet to the desired outgoing ports.

\item A large portion of the traffic is RDMA traffic send by the SmartNIC. This gives rise to new opportunities for optimization, as shall be discussed below.

\end{itemize}

We present a routing system for ML training that comprises global optimization by a centralized controller and host-controlled routing. Our system optimizes a global notion of efficiency and fairness for ML training, which we refer to as 2-layered max-min fairness by executing a simple and robust optimization algorithm with provable guarantees. We evaluate our system using a packet-level simulator, demonstrating its performance benefits. Our evaluation results show that our routing scheme outperforms traditional routing and, more importantly, almost matches the global optimum. We also contrast our methodology with recent approaches for optimizing the network topology itself via optical switching.

\section{Preliminaries}

\subsection{Training Large Models}

    \begin{table*}[t]
        \centering
        \begin{tabular}{|c|c|c|}
            \hline
            \multicolumn{3}{|c|}{LLM sizes} \\
            \hline
             Model & Num. Parameters (B) & Num. Accelerators \\
             \hline
             GPT-3~\cite{brown2020language} & 175 & 1024 GPUs \\
             OPT-175B~\cite{zhang2022opt} & 175 & 992 GPUs \\
             PaLM~\cite{chowdhery2022palm} & 540 & 6144 TPUv3 chips \\
             LLaMA2-70B~\cite{touvron2023llama} & 70 & 2048 GPUs \\
             Gopher~\cite{rae2022scaling} & 280 & 4096 TPUv4 chips \\
             MT-NLG~\cite{smith2022using} & 530 & 4480 GPUs \\
             Pangu-$\Sigma$~\cite{ren2023pangusigma} & 1000 & 512 Ascend 910 \\
             \hline
        \end{tabular}
        \caption{Table of LLMs sizes}
        \label{fig:dnn-models-table}
    \end{table*}

Training deep learning models of considerable size using extensive amounts of training data poses a challenging task.
The conventional method for training deep neural networks (DNNs) is applying stochastic gradient descent (SGD)~\cite{Kiefer1952StochasticEO}.
SGD iterations involve a selection of a random batch of training data, and updating the model's parameters (DNN weights) with respect to the error metric using backpropagation. This goes on until the model achieves the desired level of accuracy. Often, DNN training is distributed across multiple compute nodes, with each node potentially containing several GPUs. In such scenarios, network resources must be shared by different training processes, and even by different training jobs.

There are several possibilities for parallelization.
Here, we discuss the predominant parallelization methods: data parallelism and model parallelism.

\begin{figure}
    \centering
    \includegraphics[scale=0.25]{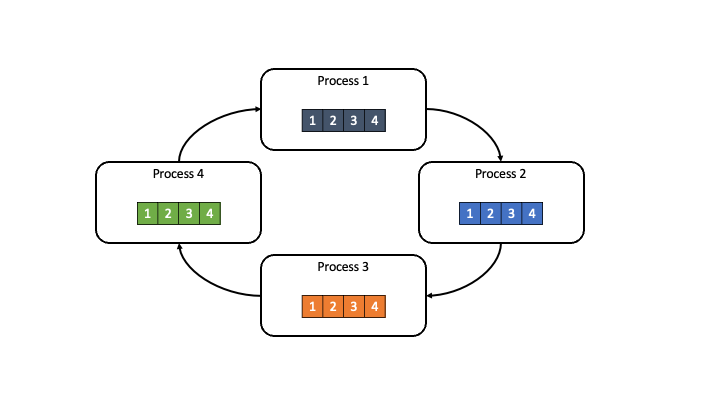}
    \caption{Example of Ring All-Reduce with 4 nodes}
    \label{fig:ring-all-reduce}
\end{figure}

Data Parallelism is a widely adopted parallelization strategy in which the dataset is divided into multiple shards, and each shard is assigned to a specific instance (e.g., GPU or host), which possesses a complete replica of the model, and conducts the forward and backpropagation steps locally. All instances synchronize their model weights during each training iteration, a process commonly knows as \emph{All-Reduce}~\cite{amodei2015deep,8868152,goyal2018accurate,Valiant:1990:BMP:79173.79181}. There are several techniques for optimizing the communication complexity of the All-Reduce process, such as ring All-Reduce~\cite{Baidu,jia2018highly,doi:10.1177/1094342005051521} (see Figure \ref{fig:ring-all-reduce}), tree All-Reduce or hierarchical ring All-Reduce~\cite{doi:10.1177/1094342005051521,8752949}.

To illustrate data parallelism, let us consider Ring All-Reduce. Suppose that there are $N$ instances associated with a training job, numbered $0$,\ldots,$N-1$. The instances are organized into a directed \emph{virtual} ring topology. To simplify exposition, suppose that each instance $i$ is preceded in the ring by instance $i-1$ modulo $N$ and followed by instance $i+1$ modulo $N$. The array of DNN link weights of each instance is partitioned into $N$ subarrays (in the exact same manner), also referred to as ``chunks''. At each iteration of the SGD, each instance updates the link weights of its locally stored model based on its assigned data shard. Then, a sequence of $2N-2$ consecutive communication rounds in which each instance transmits a chunk of the array to the next instance in the ring, is performed to generate, at each instance, the fully updated link-weight vector. See~\cite{amodei2015deep,8868152,goyal2018accurate,Valiant:1990:BMP:79173.79181} for details. In the course of these multiple communication rounds, each instance transmits data amounting to roughly twice the size of full vector of DNN link weights. 

With the introduction of large DNNs, such as LLMs, the entire model no longer necessarily fits into the memory of a single instance, or even host. Consequently, a need for dividing the model across several instances arose. This is accomplished using model parallelism~\cite{10.5555/2999134.2999271}. One type of model paralellism, called pipeline parallelism, is partitioning the model across the vertical dimension, i.e., distributing the DNN layers across different instances.
This method distributes computation across layers~\cite{DBLP:journals/corr/abs-1811-06965,harlappipedream,shoeybi2020megatronlm}. Under pipeline parallelism, in the forward pass, each instance transfers intermediate activations to the subsequent stage, whereas in the backpropagation pass, each instance conveys the gradient of the input tensor back to the preceding pipeline stage. This allows for concurrent computations by the instances, thereby accelerating training. Tensor parallelism~\cite{lepikhin2020gshard,shoeybi2020megatronlm} partitions the model horizontally, i.e., a tensor is divided into $n$ chunks, and each instance exclusively manages $\frac{1}{n}$ of the entire tensor without compromising the accuracy of the overall computation. This, however, requires additional communication between the instances.

Training large DNNs can involve hybrid parallelization strategies, combining data parallelization and both flavors of model parallelism~\cite{meta-dlrm,DBLP:journals/corr/abs-1811-06965,jia2018data,nvidia-article,naumov2020deep,shoeybi2020megatronlm}. GPT-3 training~\cite{brown2020language} is one such example. As discussed in~\cite{rajasekaran2023cassini}, the All-Reduce communication burden associated with the data parallelism aspect dominates the other communication overheads (associated with forward passes and backpropagations). We therefore focus our attention on how to optimize data transfer across the network for this crucial ingredient of the training process.

To get a sense of the communication burden of training, consider the model Bloom~\cite{workshop2023bloom} with similar architecture and number of parameters to GPT-3, using hybrid parallelization, where one full copy of the DNN is vertically partitioned to $12$ stages (groups of consecutive layers) and horizontally partitioned to $4$ chunks, and each copy of the model is replicated $8$ times using data parallelism (overall use of $384$ GPUs, or $48$ hosts if under each host there are $8$ GPUs). This means that each host holds $2$ consecutive layers of the model. Suppose that the $8$ hosts associated with each phase are organized in a virtual ring and Ring All-Reduce is used to communicate link weights between them (see above). Each All-Reduce phase will involve each host sending $14$ flows of size $>15Gbit$. Hence even in a non-blocking network with $100$Gbps links, the communication time will exceed $2$s.

\subsection{Training Clusters}

    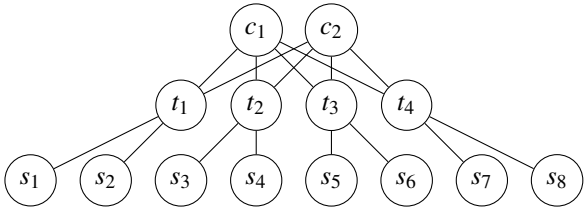
\begin{figure}
        \centering
        \begin{tikzpicture}[every node/.style={circle, draw}]
            \node (s1) at (0,-2) {$s_1$};
            \node (s2) at (1,-2) {$s_2$};
            \node (s3) at (2,-2) {$s_3$};
            \node (s4) at (3,-2) {$s_4$};
            \node (s5) at (4,-2) {$s_5$};
            \node (s6) at (5,-2) {$s_6$};
            \node (s7) at (6,-2) {$s_7$};
            \node (s8) at (7,-2) {$s_8$};
            \node (t1) at (2,-1) {$t_1$};
            \node (t2) at (3,-1) {$t_2$};
            \node (t3) at (4,-1) {$t_3$};
            \node (t4) at (5,-1) {$t_4$};
            \node (c1) at (3,0) {$c_1$};
            \node (c2) at (4,0) {$c_2$};
            \draw (s1) -- (t1);
            \draw (s2) -- (t1);
            \draw (s3) -- (t2);
            \draw (s4) -- (t2);
            \draw (s5) -- (t3);
            \draw (s6) -- (t3);
            \draw (s7) -- (t4);
            \draw (s8) -- (t4);
            \draw (t1) -- (c1);
            \draw (t1) -- (c2);
            \draw (t2) -- (c1);
            \draw (t2) -- (c2);
            \draw (t3) -- (c1);
            \draw (t3) -- (c2);
            \draw (t4) -- (c1);
            \draw (t4) -- (c2);
        \end{tikzpicture}
        \caption{Simple 2-Layer Clos network example}
        \label{fig:sample-dc-te-example}
    \end{figure}

To accommodate the fast growing compute and network demands imposed by ML workloads, many cloud providers offer ML training platforms to enable developers to use ML technologies in an efficient, flexible, and simplified way. For example, Amazon AWS offers \textbf{EC2 UltraClusters}~\cite{amazon-ec2-ultraclusters} with more than $4,000$+ NVIDIA A100 GPUs in $500+$ hosts (with each host containing $8$ GPUs).
Alibaba PAI~\cite{alibaba-pai} provides a similar service with ~$1800$ hosts and a total number of $6,000$+ GPUs. In Alibaba PAI, unlike Amazon AWS, GPU clusters can be heterogenous, with different GPU machines having very different properties.

To network hosts/GPUs, training clusters are typically organized in Clos/fat-tree topologies. See figure \ref{fig:sample-dc-te-example} for an illustration of a $2$-layer Clos network. A Clos network with $64$ top-of-rack (ToR) switches and $32$ spine switches can, for instance, support $32\times 64=2048$ GPUs/hosts (with dedicated NICs).

Training instances are mapped to GPUs/hosts according to workload placement and scheduling policies, like best fit~\cite{258979,280918,10.1145/3190508.3190517}, dot-product~\cite{10.1145/2740070.2626334,10.1145/3190508.3190517,panigrahy2011validating}, and random-fit.

\subsection{Routing Decisions Impact Training}
\label{subsec:routing-decisions}

Naturally, for training to be efficient, the vast amounts of data that must be sent across the network to accommodate the training process must be delivered quickly and efficiently, A crucial factor impacting data delivery in this context is the manner in which data flows are assigned to shortest-paths in the Clos network.

Routing in data centers is predominantly handled by spreading flows across shortest-paths using Equal-Cost MultiPath (ECMP) (ECMP) hashing~\cite{10.17487/RFC2992}. However, while ECMP traffic distribution works well (and is, in fact, optimal~\cite{chiesa2016traffic}) when flows are numerous and small, this is not the case for large flows. Since ECMP does not globally optimize flow assignment to paths, large flows can be mapped to paths traversing the same link even if this can be avoided in other route-assignments~\cite{chiesa2016traffic}. The question of how to map large (``elephant'') flows to paths in a manner that optimizes global performance has thus received considerable attention over the years~\cite{chiesa2016traffic,10.5555/1855711.1855730}.

We illustrate ECMP's deficiencies using the small network in Figure \ref{fig:sample-dc-te-example}. Suppose there are two ML training jobs: $j_1$ with dedicated hosts $(s_1,s_2,s_3)$ and $j_2$ with dedicated hosts $(s_4,s_5,s_6)$. Both training jobs are organized into virtual rings, where $j_1$ can be expressed as $(s_1 \to s_2 \to s_3 \to s_1)$, with each host sending traffic to the next host in the sequence, and $j_2$ can be expressed as $(s_4 \to s_5 \to s_6 \to s_4)$. All link capacities are $1$.

Ignoring flows with sources and destinations under the same top-of-rack (ToR), the ToR-to-ToR traffic demands induced by both jobs are expressed by the following demand matrix:

    \begin{center}
        $D = \begin{bmatrix}
            0 & 1 & 0 & 0 \\
            1 & 0 & 1 & 0 \\
            1 & 0 & 0 & 0 \\
            0 & 0 & 0 & 0
        \end{bmatrix}$
    \end{center}

As described in the matrix, there are four flows (``commodities'') that must traverse the network. Each flow can be mapped to one of four shortest paths. Observe that there are two flows leaving $t_2$ (towards different destinations). If both flows are mapped by ECMP to routes that traverse the same spine switch, say $c_1$ (and then continue to the appropriate destinations), the two flows will contend over the bandwidth of the link $(t_2,c_1)$, resulting in each transmitting traffic at a speed that is half link bandwidth. Similarly, the two flows entering $t_1$ might be mapped to paths traversing the same link. Observe that better global route configurations, where each flow sends traffic at full bandwidth, are available.

The harm to performance inflicted by ECMP's notoriously suboptimal routing configurations is further aggravated in the ML training context for two important reasons:

\begin{itemize}
    \item {\bf Flows are long.} As discussed above, flows can take a long time to complete. Hence, inefficiency resulting from co-placement of flows onto the same links can persist for a long time.

    \item {\bf Performance is dictated by the slowest flow.} Consider a Ring All-Reduce operation. Flows carry chunks of the DNN link weights array that must be sent (in a specific order) for the SGD gradient update to be completed. If traffic between one node in the ring and another is significantly slower than that between other nodes, the completion of the gradient update will be dictated by the traffic speed for that particular connection.
     
\end{itemize}

Put together, the above two points imply that slowing down a large flow as a result of suboptimal routing can both have much broader impact for the application (ML training) as a whole, and that its adverse effect can be long lasting.

\subsection{Takeaways}

\noindent{\bf Smaller networks (than all-purpose datacenters).} As reflected by the above discussion of the sizes (in terms of GPUs and hosts) of training clusters. These are considerably smaller than large-scale all-purpose datacenters. Consequently, relatively smaller networks are required to interconnect the various compute resources.

\vspace{0.1in}\noindent{\bf Predictable and persistent traffic.} Due to the repetitive structure of the training process, traffic is highly predictable. Moreover, as illustrated by our back-of-the-envelope calculations for the Ring All-Reduce example, communication is persistent, in that a single communication phase can be in the order of seconds.

\vspace{0.1in}\noindent{\bf Routing can crucially affect training speed.} The notorious limitations of traditional datacenter routing can amplified due to the specific characteristics of the training process (and, in particular, parallelization and long-lasting flows).

\vspace{0.1in}\noindent{\bf Must design for failures, and heterogeneity of compute resources, workloads, and parallelization architectures.} Different training jobs involve different hardware requirements, different parallelization strategies (e.g., All-Reduce virtual topology structures and sizes), different communication patterns, etc. This calls for robust optimization tools. 

\section{System}

We present below an overview of our objectives and a high-level description of our system design, followed by a detailed explanation.

\subsection{Design Objectives}

\noindent{\bf Optimize global communication efficiency.} Our objective is to route data across the network in a manner that optimizes global communication efficiency. We shall discuss the specific global objective targeted, which is tailored to the training context, below.

\vspace{0.1in}\noindent{\bf Fairness across jobs.} When multiple training jobs share the same infrastructure, guaranteeing fairness between jobs is important. We shall explain how this is integrated into our formal objective.

\vspace{0.1in}\noindent{\bf Support for hardware and workload diversity.} We seek solutions that naturally extend to hardware diversity, and to different ML workloads and communication patterns. We also aim for solutions that provide high efficiency also in the presence of network failures.

\subsection{Architecture}

\begin{figure*}[ht]
    \centering
    \includegraphics[width=\textwidth]{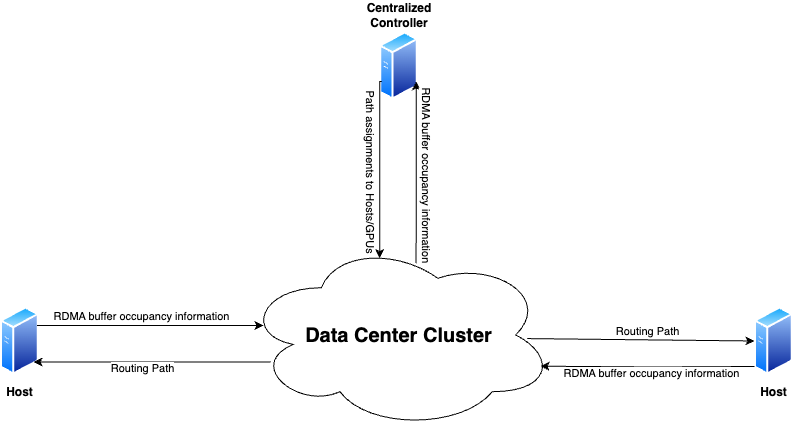}
    \caption{Hosts send RDMA buffer occupancy information and the controller sends path assignments to hosts/GPUs.}
    \label{fig:architecture-figure}
\end{figure*}

Our system architecture comprises two elements: (1) a centralized controller, and (2) the hosts. 

\vspace{0.1in}\noindent{\bf Centralized controller.} The controller is tasked with the re-assignment of routes to ML flows whenever a new flow is initiated or an ongoing flow terminates.

The input to the computation performed by the controller thus consists of the set of flows currently active in the system, specified by the two communication end-points (hosts, VMs), and the current network topology (accounting, e.g., for currently inactive links or nodes). The output specifies a path for each flow (interconnecting that flow's source and destination), and, potentially, a maximum transmission rate for every flow.

Re-computation of routes need not happen only after a new flow appears or an existing flow terminates. Indeed, an important feature of the controller is proactively computing routing configurations for the future. This capability of the controller is derived from the knowledge of how much data each flow is expected to send, as captured by the host's RDMA send buffers, and by the predictability of the traffic pattern. By knowing how much a flow intends to send upon the initialization of the flow, as well as how much bandwidth is allocated to the flow by the controller itself, the controller can estimate at what time the flow will terminate. By predicting which flow is expected to start sending next, the controller can also proactively re-compute routes prior to the start of that flow.

\vspace{0.1in}\noindent{\bf Hosts.} Hosts are tasked with enforcing the controller's decision, namely, performing source routing to ensure that each flow's traffic traverses the designated path. Hosts are also required to relay to the controller information required for informing its path computations. Specifically, whenever a new flow is initiated, the originating host is required to send to the controller the occupancy of the RDMA buffer for that flow.

\subsection{$2$-Layer Max-Min Fairness}

We next formulate the global objective targeted by our system. 

\vspace{0.1in}\noindent\textbf{The network.} We model the network as a graph $G = (V,E)$, where $V(G)$ is the set of vertices, and $E(G)$ is the set of links. $V(G)$ is the union of three disjoint sets of vertices $V(G) = V_T \cup V_C \cup V_h$, the set of ToR switches, spine switches, and Hosts, respectively. We point out that while in our formulation, communication end-points are perceived as hosts, our formulation and results can be extended to the scenario that these are GPUs or VMs.

\vspace{0.1in}\noindent\textbf{Training jobs.} There is a set of training jobs $J$, where each job in $J$ is represented as a directed graph $G_J$ over the vertex set $V_h$. Edges in this graph correspond to direct communication between hosts, as prescribed by the data parallelization strategy. For instance, two hosts, $h_1$ and $h_2$, are required to directly communicate in a Ring All-Reduce operation, this will be represented by a (directed) edge between their associated vertices in $G_J$.

\vspace{0.1in}\noindent\textbf{Path assignments.} Recall that the controller assigns new paths to flows whenever a new flow starts or an existing flow terminates. The input to each computation phase performed by the controller is thus a set of \emph{active} training jobs $J_a\subseteq J$. Since each training job description specifies pairs of communicating hosts, this input can be regarded as a set of flows (commodities) $C$, comprised of disjoint subsets of flows, where each subset $C_j$ is associated with unique training job $j\in J$. Let $P_{uv}$ be the set of paths interconnecting the hosts $u$ and $v$ in $V_h$. A path assignment maps each commodity in $C_j$ to a path $p\in P_{uv}$, where $u$ and $v$ are the source and destination hosts for this commodity. A path assignment induces a bandwidth share $b_c\geq 0$ for each commodity $c$, where $b_c$ specifies the speed at which commodity $c$ sends traffic along its designated path in the equilibrium reached by the congestion control dynamics. 

\vspace{0.1in}\noindent\textbf{Global optimization objective.} Consider an individual training task $j$ and its associated set of commodities, $C_j$. Due to the All-Reduce traffic pattern derived by data parallelism, the overall running time of each communication phase is dependent on the \emph{slowest} connection. Hence, the optimizing performance with respect to a single job involves maximizing the bandwidth consumed by the slowest flow. Consider, now, the allocation of bandwidth across different jobs. We strive to optimize the max-min fairness notion of optimality. The rationale is that fast speed for one job should not come at the expense of harming a slower job. Since each job's performance is influenced by its slowest connection, this amounts to maximizing the bandwidth share allocated to the slowest connection \emph{across all jobs}, i.e., 

\[ maximize\ min_{j\in J}min_{c\in C_j}\ b_c = maximize\ min_{c\in C}\ b_c\]

This can be regarded as a $2$-layer max-min fairness objective, where max-min fairness is applied both across the commodities corresponding to the same job and across jobs.

\vspace{0.1in}\noindent\textbf{Remark: nonblocking \emph{vs.} oversubscribed networks.} We do not make the assumption that the network is nonblocking and so our solutions also target scenarios where not all permutation traffic matrices between communication end-points can be simultaneously supported. We emphasize that even if the network itself is nonblocking, oversubscription can arise due to:

\begin{itemize}

\item {\bf Network failures.} When links of nodes fail, the network, even if initially nonblocking, can naturally become oversubscribed. We envision the controller's recomputation of paths as also triggered by such events.

\item {\bf GPU:NIC ratio < 1.} In some training platforms, e.g., those leveraging Amazon EC2 P4d instances~\cite{amazon-ec2-ultraclusters}, the number of GPUs in a server might exceed the number of NICs (e.g., $8$ GPUs and $4$ NICs~\cite{amazon-ec2-ultraclusters,nvidia-dgx-a100}). This naturally induces inherent oversubscription (since even the $8$ GPUs within the same server cannot simultaneously send traffic at line rate.
\end{itemize}

Our methods apply to all above discussed scenarios.

\subsection{Controller Design Details}

\subsubsection{\bf Route-Optimization Scheme}

The greedy algorithm

    \label{alg:mcvlc}
    \begin{algorithm}
        \scriptsize
        \DontPrintSemicolon
        \KwIn{$C = \cup_{j \in J_a} C_j$ set of commodities associated with active training jobs}
        \KwOut{Path assignments $P$}
        $P \gets \{\}$ \Comment{Mapping of commodity to routing path}
        $L \gets \{\}$ \Comment{Mapping of link to number of path assignments traversing it}
        \For{$c=(u,v) \in C$}{
            \Comment{Initially assign the path with the lowest spine id, where $P_{uv}$ is ordered by spine ids}
            $p \gets P_{uv}(0)$\;
            $l_{c} \gets \max_{e:e \in p} L(e)$\; \Comment{$l_c$ Counts the number of assignments traversing through the "busiest" link of path $p$}
            \Comment{Loop through all paths and change to the one with least number of assignments}
            \For{$p' \in P_{uv}$}{
                $l_{c}' \gets \max_{e:e \in p'} L(e)$\;
                \If{$l_{c}' < l_{c}$}{
                    $l_{c} \gets l_{c}'$\;
                    $p \gets p'$
                }
            }
            $L(e) \gets L(e) + 1 \forall e \in p$\;
            $P \gets P \cup \{c \to p\}$\;
        }
        \Return{$P$}\;
        \caption{{\sc Greedy path assignments}}
    \end{algorithm}

To accommodate timely and effective decisions, the algorithm executed by the controller should be both fast to run and provide close-to-optimal quality solutions. The pseudocode for our algorithm is presented above. The greedy algorithm goes over commodities in some arbitrary order, mapping each commodity to the ``least congested'' path so far. As we show below, this simple and robust algorithm yields solutions that are both provably and empirically close to those of the optimal solution.

\subsubsection{\bf Forward-Looking Optimization}

\noindent{\bf Predicting future traffic patterns.} Recall that controller decisions (path assignments) are enforced whenever a new flow enters the system or an existing flow terminates. To avoid performing computations only after such an event occur, and the associated time lag, the controller attempts to predict how the set of active flows might change and precompute a path assignment for the predicted scenario. Determining which flows will be active requires figuring out when the existing flows are expected to terminate and when new communication will begin. We discuss each of these challenges separately.

\begin{itemize}
\item {\bf Estimating when the current flows will terminate.} To estimate when a given flow terminates, we need to know how much data it is expected to send and at which speed the data will be sent. The former can be derived from hosts' reports of the occupancy of the RDMA send buffer whereas the latter is dependent on the path assignment produced by the controller itself.

\item {\bf Estimating when a new flow will join the system.} The traffic patterns induced by large ML model training are rather predictable. Hence, predictive models can be leveraged to estimate when currently inactive flows shall become active.

\end{itemize}

\noindent{\bf Precomputing path assignments for failure scenarios.} To protect against node/link failures, the controller could perform path-assignment optimizations for different failure scenarios (e.g., all single link failures) in the background. When such a failure occurs, the controller could simply issue its precomputed path assignment, avoiding the online optimization runtime.

\subsubsection{\bf Runtime Parallelization Strategies}

While, ideally, controller computations will occur before the arrival/departure of flows, scenarios where the controller failed to correctly predict the upcoming traffic conditions call for fast computation of new path assignments after the fact. We point out that the greedy path-assignment algorithm, which is already expected to be fast due to its simple nature, can be parallelized as follows:

\begin{enumerate}

\item Create a graph whose vertices are the commodities and two vertcies (commodities) are connected by an edge if and only if they share the same source ToR or destination ToR.

\item Divide the graph into maximal connected components.

\item Apply the algorithm to the commodities in each connected component in parallel.
\end{enumerate}

\subsection{Host Design Details}

\subsubsection{Overview}
Hosts are connected to the network using SmartNICs (Smart Network Interface Cards) that provide hardware offloading for RDMA (Remote Direct Memory Access). The SmartNIC has direct access to the host memory, and thus data transfers are handled by the SmartNIC without the CPU's intervention. 

There are two notable assumptions with regards to the host behavior. Firstly, the assumption is that the lion's share of the traffic is RDMA traffic, which is sent by the SmartNIC. Therefore, when an RDMA message is invoked the SmartNIC is aware of the length of the RDMA message. While there may be non-RDMA traffic in the network, it is assumed that this is a small fraction of the network bandwidth. 
The second assumption is that the SmartNIC controls the path through which a packet is forwarded. 

The scheme that was described above requires the host to be able to detect new flows and update the controller about them and to route these flows (RDMA messages) according to the controller's policy. Our assumption is that hosts use the controller's routing policy to determine how ``elephant'' flows are forwarded, i.e., large RDMA messages, whereas ``mice'' flows are forwarded using conventional load balancing approaches such as ECMP.

\subsubsection{Leveraging the RDMA Message Length}
The length of an RDMA transfer is known when its Working Queue Element (WQE)~\cite{IB} is created. This allows a host to distinguish elephants from mice, as only elephant flows require the controller's routing decision. The host notifies the controller about each new elephant and its length, allowing the controller to predict when the message will be completed.

\subsubsection{Enforcing route decisions}
The host can control the routing decision by using segment routing (e.g., SRv6~\cite{rfc8402}). However, since segment routing is not commonly deployed in data center networks, a more common approach for the host to control the selected path is by controlling the content of the header fields used by the switches for making ECMP-based decisions. 

Specifically, when using RoCEv2~\cite{RoCEv2}, the UDP source port is not used by the receiver's UDP layer, and therefore an RDMA sender can vary the UDP source port in order to use multiple paths. Since a data center network is a contained administrative domain, the switch load balancing mechanisms can be configured in a way that is consistent with this assumption. 

Our scheme allows a straightforward way of enforcing routing decisions: the controller notifies a host about the selected path for a given flow by assigning a specific UDP source port for it. This approach does not require hosts to be aware of the network topology, while the controller can assign UDP source port values to different flows from a network-wide perspective.

\subsubsection{New flow procedure}
The procedure taken by the host when a new (elephant) flow appears is as follows:

\begin{itemize}
	\item A new RDMA message needs to be transmitted. If the message length is less than a predetermined threshold, it is forwarded using ECMP (and the procedure is concluded).
	\item Otherwise, a request is sent to the controller, containing details about the new message, including its length.
	\item The controller assigns a path (UDP source port), according to its (online computed or precomputed) path assignment, for the new flow, and sends a response to the host.
	\item The host transmits the flow (message) along the assigned path.
\end{itemize}

\subsubsection{Safety}
Safety mechanisms could be incorporated into hosts to handle unexpected changes to the traffic pattern or network topology (while the controller is processing these changes). Such scenarios include a newly arrived job starting to send and node/link failure. The host could run measurements on its assigned paths and revert to ECMP when experienced latency/pack loss rate for a certain commodity exceed a certain threshold. Another potential safety mechanism can take into account the response time from the controller. If the controller's response to new flow assignment becomes prohibitively high the host can fall back into a default path assignment scheme such as ECMP. It should be noted that safety mechanisms that run at the host are independent of any safety mechanisms that run at the controller. The path assignment mechanism as well as the flow prediction mechanism can each have it own safety mechanism which allows the controller to fall back to a simpler default algorithm.

\section{Theoretical Evaluation}

We consider our 2-layered max-min-fairness objective in the context of a \textit{directed} 2-layer Clos networks.

\noindent{\bf Input:} A 2-layer directed Clos network (full bipartite graph with bidirectional links) $G=(V,E)$, where $V=T\bigcup S$ and link have uniform capacities (to simplify exposition, and WLOG, all link capacities are assumed to be $1$ henceforth).
The two disjoint vertex sets, $T$ and $S$ represent the ToRs and the spine switches, respectively, whereas the edge set $E$ represents communication links).
The input also consists of a superset of commodities $C$ of the form $(u,v)\in T^2$, each associated with demand $d_{uv}\in\{0,1\}$. (We ignore the host-to-ToR for the purpose of this analysis and zoom in on the inter-ToR efficiency).

We consider the simple greedy algorithm $Greedy$ that goes over the commodities in $C$ in some arbitrary order and maps each commodity to the least congested shortest-path interconnecting its source and destination.

We prove the following theorem:

\begin{theorem}
$Greedy$ provides a $2$-approximation to the $2$-layered max-min fairness objective in directed $2$-layer Clos networks.
\end{theorem}

The proof of the above theorem crucially hinges on the correctness of the next result.

\begin{theorem}\label{thm:greedy-approx}
The maximum number of paths traversing a link in $Greedy$'s outputted path assignment is at most $2$ times higher than this number in the optimal solution.
\end{theorem}

\begin{proof}
Fix an arbitrary order over the commodities in $C$ and suppose that $Greedy$ considers commodities according to this order.
Let $OPT_j$ and $Greedy_j$ be the assignments of shortest-paths to the first $j$ commodities in the optimal solution (for all commodities) and the greedy-outputted solution, respectively.
To prove the theorem, we show for every commodity $k$ for which $|Greedy_k|>|Greedy_{k-1}|$, it holds that $|Greedy_k|\leq 2|OPT_k|$.
Observe that this indeed implies the theorem.

Consider a specific commodity $k=(u,v)\in C$ for which the above holds.
Let $P_{uv}$ be the set of shortest-paths between $u$ and $v$ in $G$.
For each path $p\in P_{uv}$ let the \textit{congestion} of the path be the maximum congestion (aggregate demand) across its (two) links.
Observe that the congestion of each path $p\in P_{uv}$ in $Greedy_{k-1}$ must be exactly $|Greedy_k|-1$ (for otherwise, $Greedy_k$ would have been lower or $|Greedy_k|=|Greedy_{k-1}|$ --- a contradictions).
Let $L_1$ be the set of $|S|$ links (edges) leaving the vertex $u$ and $L_2$ be the set of $|S|$ links entering the vertex $v$. Observe that each path $p\in P_{uv}$ traverses a unique link in $L_1$ and a unique link in $L_2$.
Since the congestion on each path is $|Greedy_k|-1$, either its $L_1$ link or its $L_2$ link, or both, must experience congestion of $|Greedy_k|-1$ in $Greedy_{k-1}$.
Since all paths in $P_{uv}$ have congestion $|Greedy_k|-1$, at least half of the links in $L_1$, or at least half of the links in $L_2$ (or both), must have congestion $|Greedy_k|-1$.
WLOG, suppose that this is the case for $L_1$.
This implies that the aggregate demand from $u$ to other vertices in the first $k-1$ commodities is $\frac{|S|}{2}(|Greedy_k|-1)$.
In any solution, including $OPT_{k-1}$, this aggregate demand must be split across the $|S|$ links in $L_1$.
We analyze two cases:

\begin{itemize} 

\item {\bf All links in $L_1$ have congestion $\frac{|Greedy_k|-1}{2}$ in $OPT_{k-1}$.} Observe that in this scenario, the load of at least one link in $L_1$ in $OPT_k$ must be $\frac{|Greedy_k|-1}{2}+1=\frac{|Greedy_k|+1}{2}$ and so $Greedy_k\leq 2OPT_k$.

\item {\bf At least one link in $L_1$ has congestion strictly higher than $\frac{|Greedy_k|-1}{2}$ in $OPT_{k-1}$.} Since congestion is integral (as demands are in $\{0,1\}$), this implies that the congestion on that link is at least $\frac{|Greedy_k|-1}{2}+\frac{1}{2}=\frac{|Greedy_k|}{2}$. Hence, in this scenario, too, $Greedy_k\leq 2OPT_k$.

\end{itemize}
 
\end{proof}

To see why the above implies our max-min-fairness result, consider let $x$ be the maximum number of paths traversing a link in the final solution for $Greedy$. By the above result, this number for $OPT$ cannot be less than $\frac{x}{2}$. Hence, the minimum bandwidth for a commodity in $Greedy$ is $\frac{1}{x}$, whereas in $OPT$ there is a commodity with bandwidth at most $\frac{2}{x}$. The theorem follows.

\section{Theoretical Evaluation}
\label{sec:evaluation}

Our evaluation answers the following questions: 
\textbf{(1)} How does our framework compare with various alternative routing schemes and with the global optimum? (We contrast our proposal with TopoOpt in the next section), \textbf{(2)} How does our scheme perform under network faults? \textbf{(3)} How long does our approach take to assign routes? \textbf{(4)} How does our scheme perform under network oversubsciption?

\subsection{Methodology}
\label{subsec:methodology}
We evaluate our system under various workloads, and with respect to different performance objectives with a custom simulator.

\vspace{0.1in}\noindent\textbf{Topology.} We consider a $2$-level Clos network with $32$ spine switches and $64$ ToR switches. Each host in the network holds $8$ GPUs equipped with a dedicated NIC, similar to NVidia DGX A100~\cite{nvidia-dgx-a100}. Each link has capacity of $100$ Gbps.

\begin{table*}[t]
    \centering
    \begin{tabular}{|c|c|c|c|}
        \hline
         Model & Tensor Parallelism & Pipeline Parallelism & Full Copy \\
         \hline
         BLOOM~\cite{workshop2023bloom} & 4 & 12 & 48 \\
         GPT-3~\cite{brown2020language} & 8 & 8 & 64 \\
         LLaMA2-70B~\cite{touvron2023llama} & 8 & 16 & 128 \\
         \hline
    \end{tabular}
    \caption{Dimension Sizes (Number of GPUs)}
    \label{fig:evaluation-dimension-sizes}
\end{table*}

\vspace{0.1in}\noindent\textbf{DNN Workloads.} We consider three real-world DNN models: BLOOM~\cite{workshop2023bloom}, GPT-3~\cite{brown2020language} and LLaMA2-70B~\cite{touvron2023llama}. Each model is trained with $3D$ parallelism (Tensor, Pipeline and Data parallelism), where in Table \ref{fig:evaluation-dimension-sizes} we state the sizes of the tensor and pipeline parallelism. The data parallelism dimension size is an hyperparameter is chosen from the set $\{2,4,8\}$. Each job uses Ring All-Reduce for data parallelization.
Model weights, gradients and optimizer state are stored in full-precision (float32 - $4$ bytes per parameter).
We calculate the estimated forward and backpropagation time as in ~\cite{brown2020language}.
Similar to ZeRO~\cite{rajbhandari2020zero} and PyTorch Fully-Sharded Data-Parallelism (FSDP)~\cite{zhao2023pytorch}, we simulate the communication between the GPUs that store the same subset of model parameters.

\vspace{0.1in}\noindent\textbf{Job Arrival Plan.} In our simulations, each job is randomly assigned to a subset of the hosts in the cluster, and arrives at a different time, drawn according to the uniform distribution over a period of $10$ seconds.
The simulation starts when the first job arrives, where for each training job we simulate $10$ training iterations, each divided into computation phase (forward and backpropagation step) and communication phase (parameter synchronization phase using Ring All-Reduce).
The number of concurrent training jobs varies from $1$ training job to $5$ concurrent jobs (with each job randomly being a BLOOM, GPT-3 or LLaMA2-70B training job).

\vspace{0.1in}\noindent\textbf{Comparables} to our routing scheme include: (1) \textbf{ECMP}, (2) \textbf{Edge-Coloring}~\cite{10.1145/2619239.2626309}, where the per-training-phased optimization problem is cast as a $k$-edge-coloring problem and solved via a combinatorial algorithm, (3) \textbf{Simulated Annealing}, as suggested in Hedera~\cite{10.5555/1855711.1855730}, and (4) \textbf{ILP optimum}~\cite{10.1145/2619239.2626309}, which is the scheme that exactly optimizes the $2$-layered max-min-fair objective in each training iteration by solving an ILP.

\vspace{0.1in}\noindent\textbf{Infrastructure:}
We ran our experiments on a MacBook Pro with Apple M1 Max $10$ Core CPU, $32$ Core GPU, $16$ Core Neural Engine and $32$GB RAM.

\vspace{0.1in}\noindent{\bf Code:} To support the reproducibility of our results and further investigation we release the code for our algorithm and simualation framework. The code is available at~\cite{DCN-Simulators}.

\begin{figure*}
    \captionsetup{justification=centering}
    \centering
    \begin{subfigure}{0.49\textwidth}
        \includegraphics[width=1.0\textwidth]{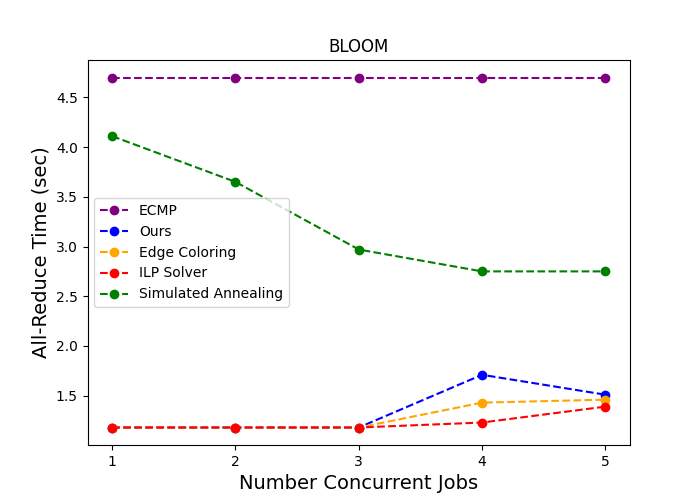}
        \caption{Ring size 2}
        \label{fig:all-reduce-time-ring-2}
    \end{subfigure}
    \hfill
    \begin{subfigure}{0.49\textwidth}
        \includegraphics[width=1.0\textwidth]{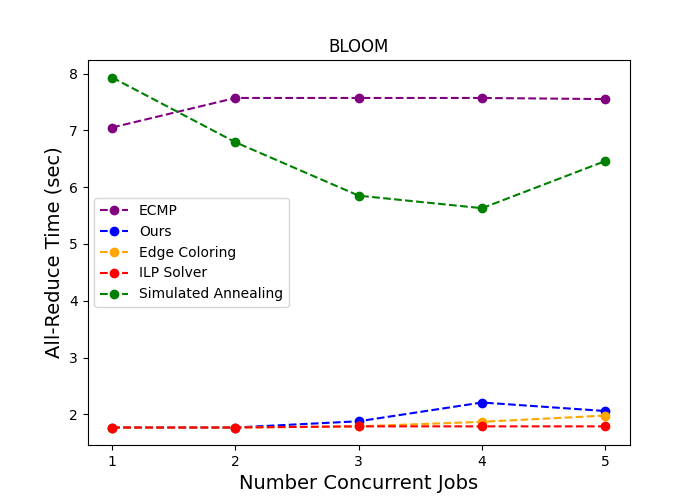}
        \caption{Ring size 4}
        \label{fig:all-reduce-time-ring-4}
    \end{subfigure}
    \hfill
    \begin{subfigure}{0.49\textwidth}
        \includegraphics[width=1.0\textwidth]{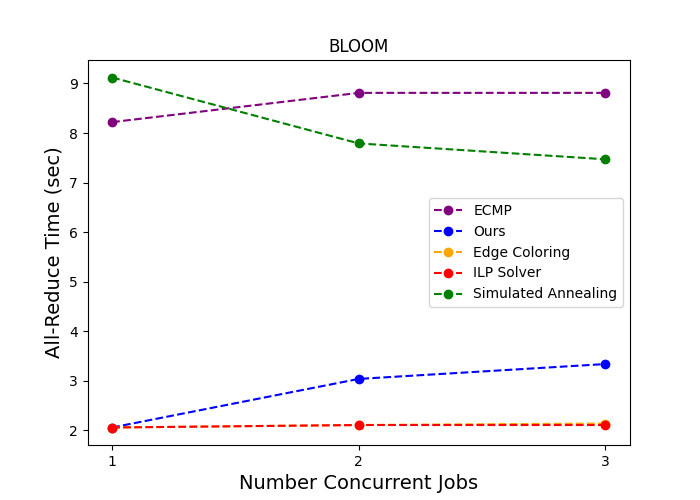}
        \caption{Ring size 8}
        \label{fig:all-reduce-time-ring-8}
    \end{subfigure}
    \hfill
    \begin{subfigure}{0.49\textwidth}
        \includegraphics[width=1.0\textwidth]{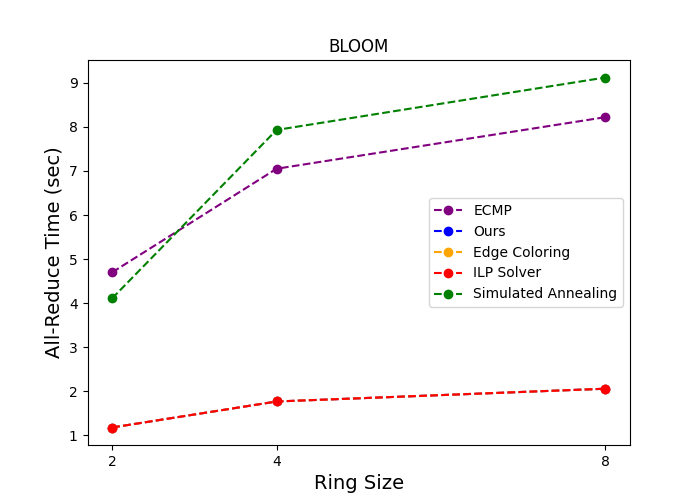}
        \caption{1 Concurrent Jobs}
        \label{fig:all-reduce-time-1-jobs}
    \end{subfigure}
    \hfill
    \begin{subfigure}{0.49\textwidth}
        \includegraphics[width=1.0\textwidth]{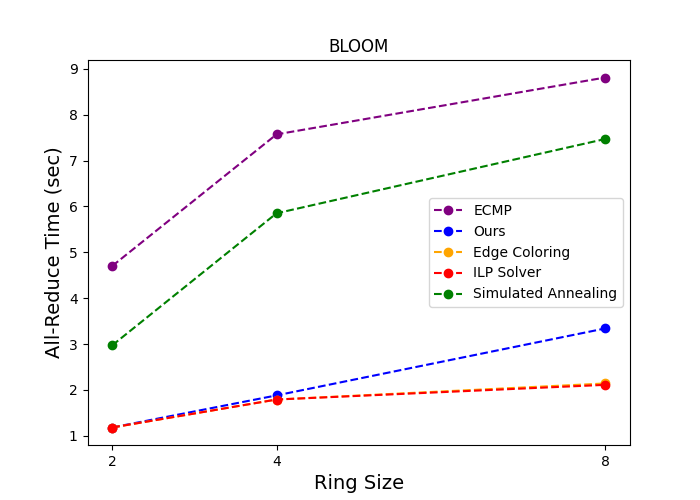}
        \caption{3 Concurrent Jobs}
        \label{fig:all-reduce-time-3-jobs}
    \end{subfigure}
    \hfill
    \begin{subfigure}{0.49\textwidth}
        \includegraphics[width=1.0\textwidth]{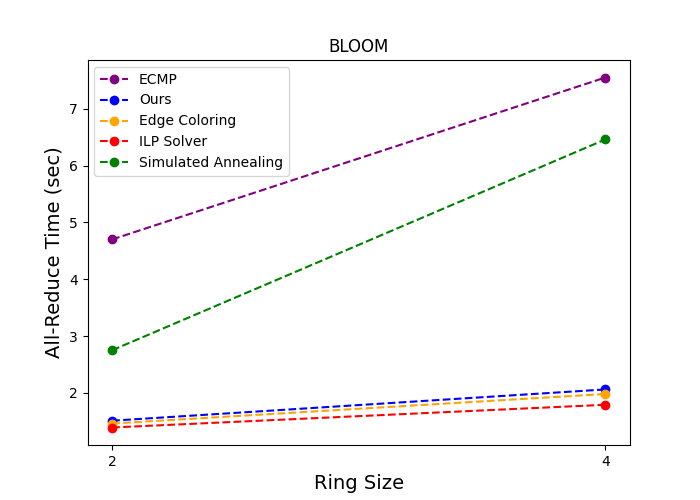}
        \caption{5 Concurrent Jobs}
        \label{fig:all-reduce-time-5-jobs}
    \end{subfigure}
    \caption{All-Reduce (Parameter Synchronization) time of Bloom~\cite{workshop2023bloom}. In figures \ref{fig:all-reduce-time-ring-2} - \ref{fig:all-reduce-time-ring-8} we vary the number of concurrent jobs submitted in the cluster on the x-axis, and in figures \ref{fig:all-reduce-time-1-jobs} - \ref{fig:all-reduce-time-5-jobs} we vary the number of All-Reduce ring size on the x-axis.}
    \label{fig:all-reduce-time}
\end{figure*}


\begin{figure*}
    \captionsetup{justification=centering}
    \centering
    \begin{subfigure}{0.49\textwidth}
        \includegraphics[width=1.0\textwidth]{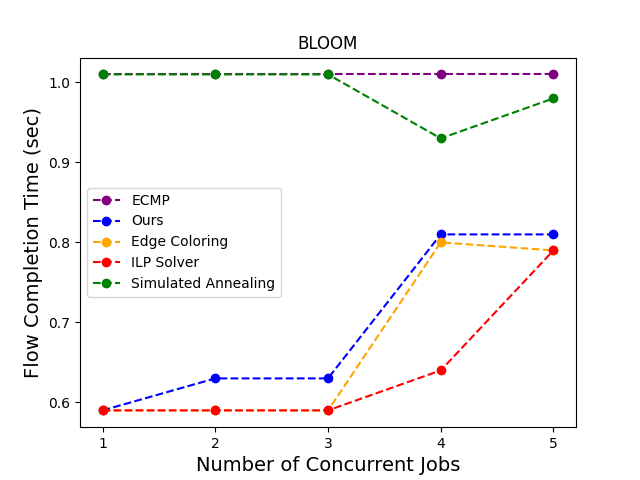}
        \caption{Ring size 2}
        \label{fig:fct-bloom-ring-2}
    \end{subfigure}
    \hfill
    \begin{subfigure}{0.49\textwidth}
        \includegraphics[width=1.0\textwidth]{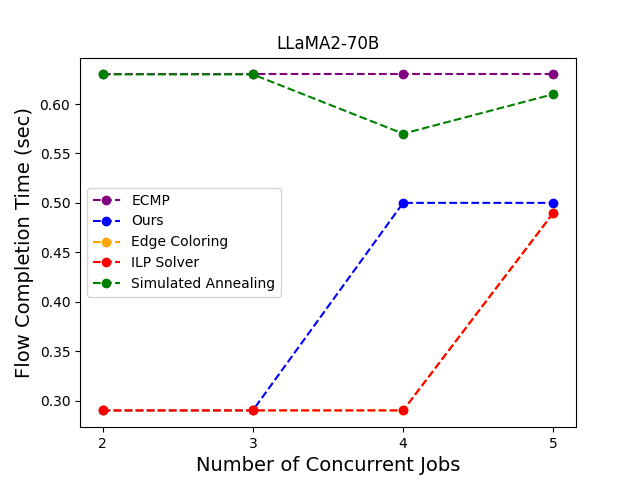}
        \caption{Ring size 2}
        \label{fig:fct-llama2-ring-2}
    \end{subfigure}
    \hfill
    \begin{subfigure}{0.49\textwidth}
        \includegraphics[width=1.0\textwidth]{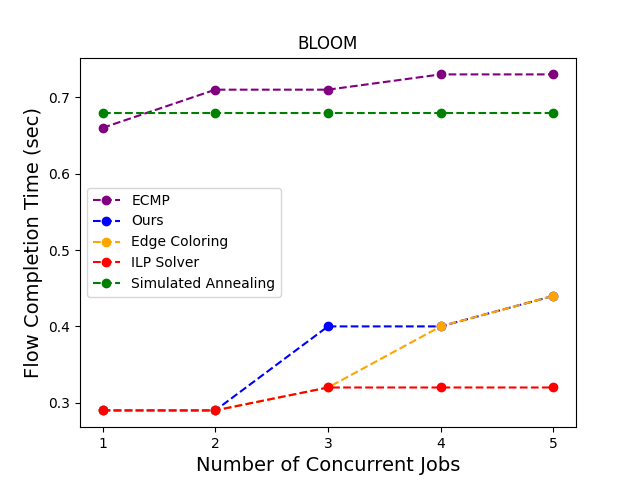}
        \caption{Ring size 4}
        \label{fig:fct-bloom-ring-4}
    \end{subfigure}
    \hfill
    \begin{subfigure}{0.49\textwidth}
        \includegraphics[width=1.0\textwidth]{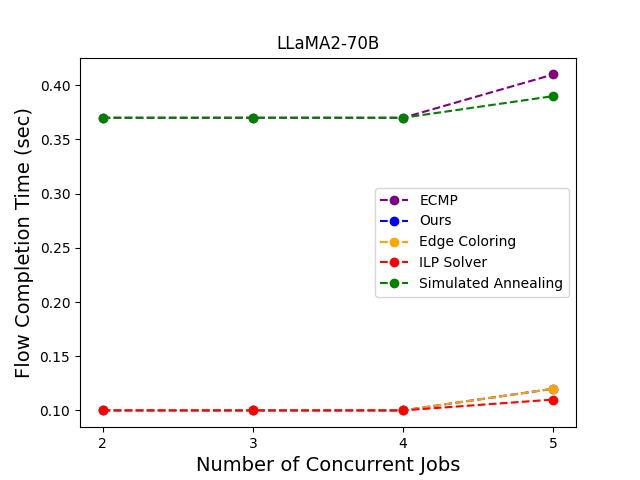}
        \caption{Ring size 4}
        \label{fig:fct-llama2-ring-4}
    \end{subfigure}
    \hfill
    \begin{subfigure}{0.49\textwidth}
        \includegraphics[width=1.0\textwidth]{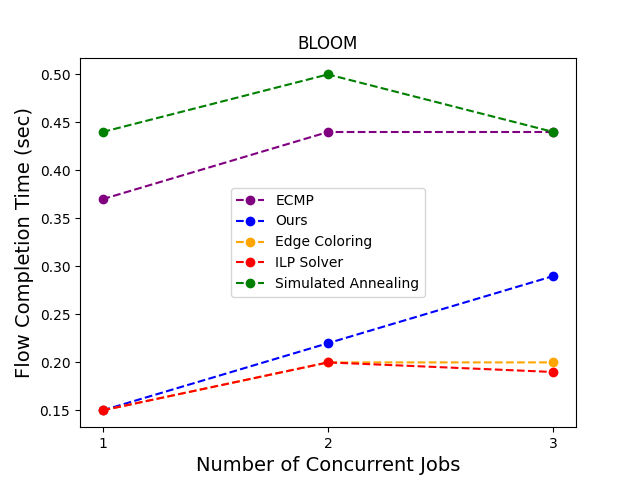}
        \caption{Ring size 8}
        \label{fig:fct-bloom-ring-8}
    \end{subfigure}
    \hfill
    \begin{subfigure}{0.49\textwidth}
        \includegraphics[width=1.0\textwidth]{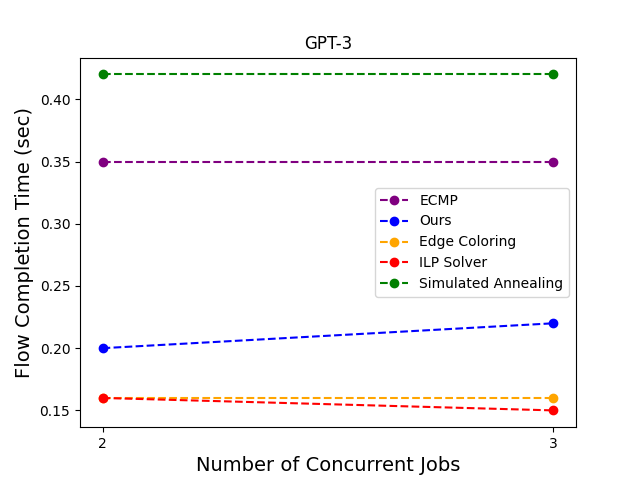}
        \caption{Ring size 8}
        \label{fig:fct-gpt3-ring-8}
    \end{subfigure}
    \caption{Flow Completion Time (FCT) of various architectures. In figures \ref{fig:fct-bloom-ring-2} - \ref{fig:fct-gpt3-ring-8} we vary the number of concurrent jobs submitted in the cluster on the x-axis.}
    \label{fig:flow-completion-time}
\end{figure*}


\begin{figure*}
    \captionsetup{justification=centering}
    \centering
    \begin{subfigure}{0.49\textwidth}
        \includegraphics[width=1.0\textwidth]{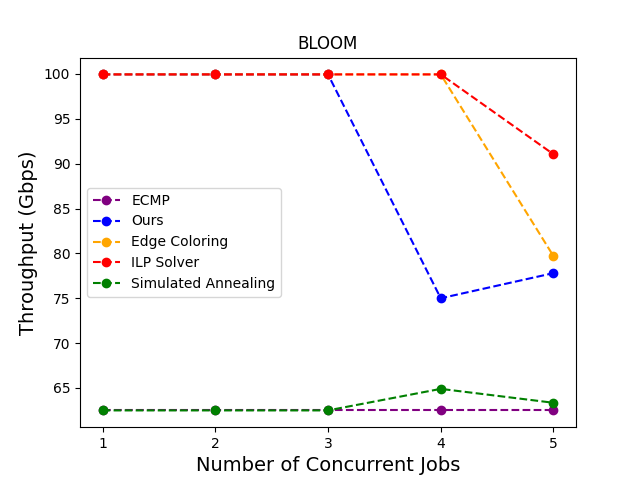}
        \caption{Ring size 2}
        \label{fig:throughput-bloom-ring-2}
    \end{subfigure}
    \hfill
    \begin{subfigure}{0.49\textwidth}
        \includegraphics[width=1.0\textwidth]{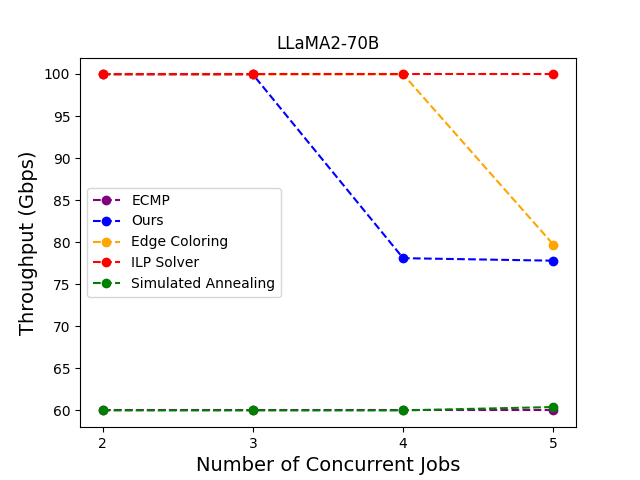}
        \caption{Ring size 2}
        \label{fig:throughput-llama2-ring-2}
    \end{subfigure}
    \hfill
    \begin{subfigure}{0.49\textwidth}
        \includegraphics[width=1.0\textwidth]{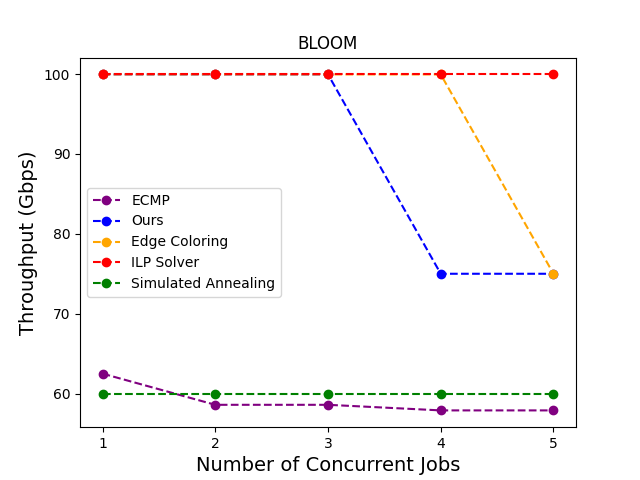}
        \caption{Ring size 4}
        \label{fig:throughput-bloom-ring-4}
    \end{subfigure}
    \hfill
    \begin{subfigure}{0.49\textwidth}
        \includegraphics[width=1.0\textwidth]{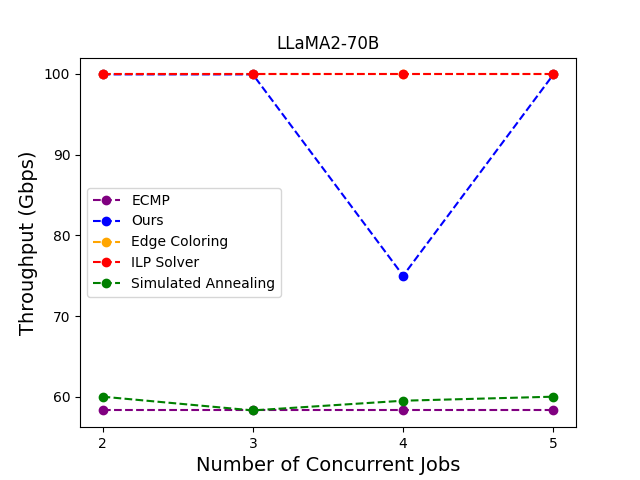}
        \caption{Ring size 4}
        \label{fig:throughput-llama2-ring-4}
    \end{subfigure}
    \hfill
    \begin{subfigure}{0.49\textwidth}
        \includegraphics[width=1.0\textwidth]{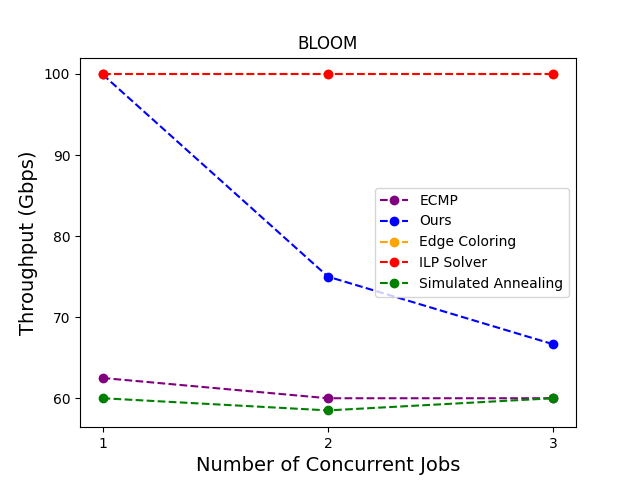}
        \caption{Ring size 8}
        \label{fig:throughput-bloom-ring-8}
    \end{subfigure}
    \hfill
    \begin{subfigure}{0.49\textwidth}
        \includegraphics[width=1.0\textwidth]{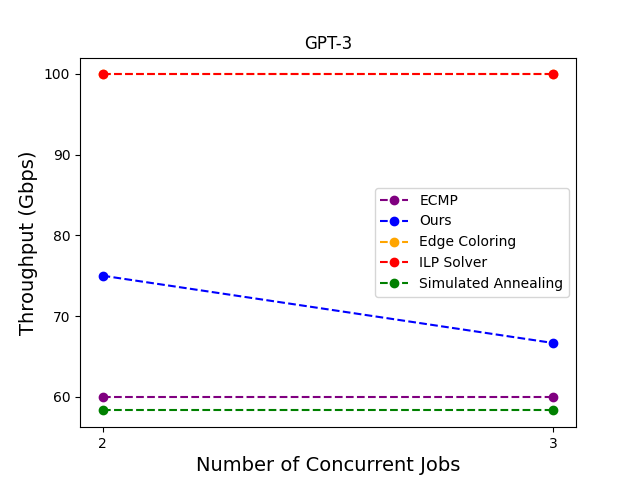}
        \caption{Ring size 8}
        \label{fig:throughput-gpt3-ring-8}
    \end{subfigure}
    \caption{Throughput of various architectures. In figures \ref{fig:throughput-bloom-ring-2} - \ref{fig:throughput-gpt3-ring-8} we vary the number of concurrent jobs submitted in the cluster on the x-axis.}
    \label{fig:throughput}
\end{figure*}

\begin{figure*}[h]
    \centering
    \includegraphics[width=\textwidth]{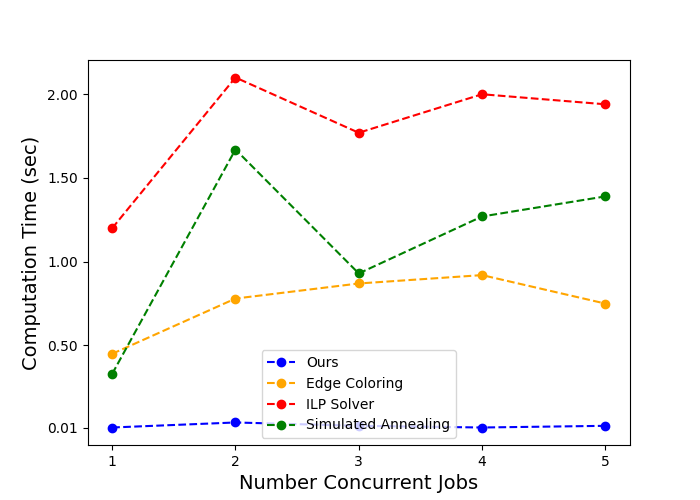}
    \caption{Measuring the computation time of routing schemes}
    \label{fig:computation-time-measurement}
\end{figure*}

\vspace{0.1in}\noindent\textbf{Metrics:}
We consider several objectives: 

\begin{itemize}
    \item {\bf All-Reduce time.} The All-Reduce time for a job is computated as follows. For each training iteration of the job, the time for the longest flow between GPUs that hold the same set of model parameters is computed. These values are then averaged across all training iterations. 
    \item The average {\bf flow completion time (FCT)} across all flows in all training iterations for a given job. 
    \item The average {\bf throughput} across all flows in all training iterations for a given job. 
    \item {\bf Runtimes (latency)} of the different evaluated schemes.
\end{itemize}

\noindent\vspace{0.1in}\textbf{Fault model.}
To examine performance under network faults, we randomly bring down $1/4/8$ spine switches. We then contrast the performance of our routing scheme and the alternatives with respect to the above discussed metrics.
\newline
\newline
\noindent Importantly, the number of active hosts, the set of jobs, the mapping of jobs to hosts, and the data parallelization strategy, are the same in any specific comparison.

\subsection{Performance Comparison}

In Figures \ref{fig:all-reduce-time} - \ref{fig:throughput}, we present representative results for the All-Reduce time, flow completion time, and throughput (respectively). Each sub-figure presents the performance with respect to a single job for a specific choice of parameters (DNN model, number of concurrent jobs in the system, data-parallelization strategy). Each data point in the figures captures an average of $10$ training iterations. Observe that our scheme comes very close to the optimum (as captured by the combinatorial Edge Coloring algorithm and the ILP solution) in terms of the All-Reduce time, our primary performance metric. This is despite sometimes exhibiting somewhat lower \emph{average} performance than the optimum for FCT and throughput (though still being superior to Simulated Annealing and ECMP). We believe that the reason for this is that because the duration of the training phase is determined by the \emph{slowest} connections, the average performance across flows of Edge Coloring and ILP are not translated to significant improvements in performance.

\begin{figure*}
    \captionsetup{justification=centering}
    \centering
    \begin{subfigure}{0.49\textwidth}
        \includegraphics[width=1.0\textwidth]{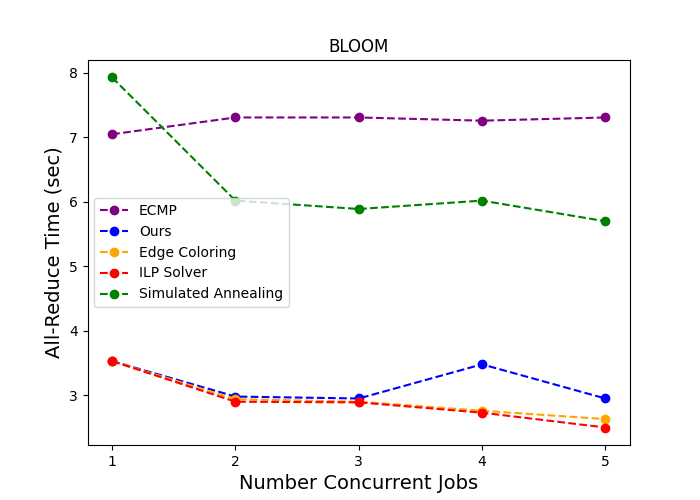}
        \caption{1 Core Failures, Ring size 4}
        \label{fig:all-reduce-time-ring-4-core-failures-1}
    \end{subfigure}
    \hfill
    \begin{subfigure}{0.49\textwidth}
        \includegraphics[width=1.0\textwidth]{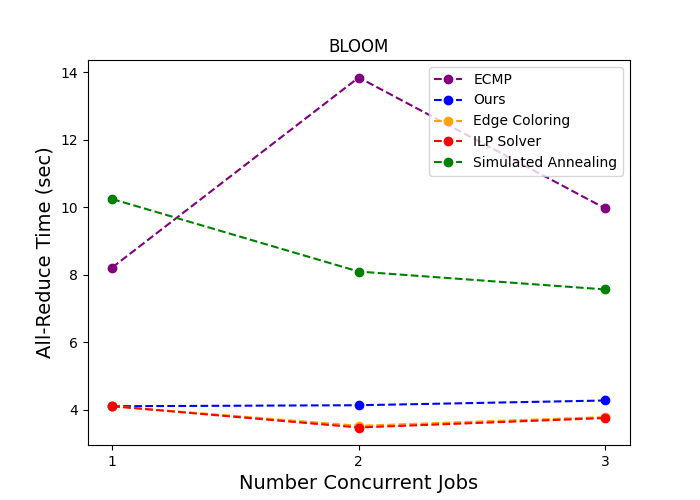}
        \caption{1 Core Failures, Ring size 8}
        \label{fig:all-reduce-time-ring-8-core-failures-1}
    \end{subfigure}
    \hfill
    \begin{subfigure}{0.49\textwidth}
        \includegraphics[width=1.0\textwidth]{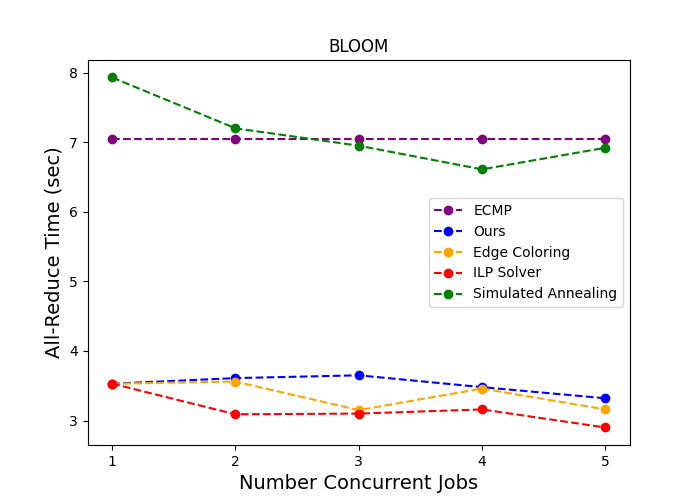}
        \caption{4 Core Failures, Ring size 4}
        \label{fig:all-reduce-time-ring-4-core-failures-4}
    \end{subfigure}
    \hfill
    \begin{subfigure}{0.49\textwidth}
        \includegraphics[width=1.0\textwidth]{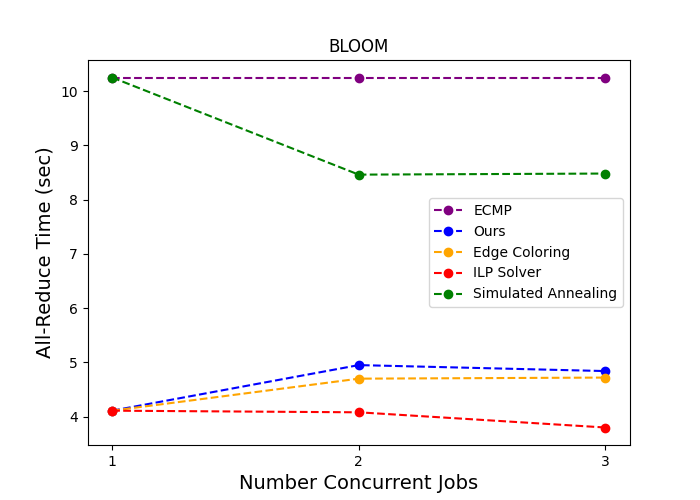}
        \caption{4 Core Failures, Ring size 8}
        \label{fig:all-reduce-time-ring-8-core-failures-4}
    \end{subfigure}
    \hfill
    \begin{subfigure}{0.49\textwidth}
        \includegraphics[width=1.0\textwidth]{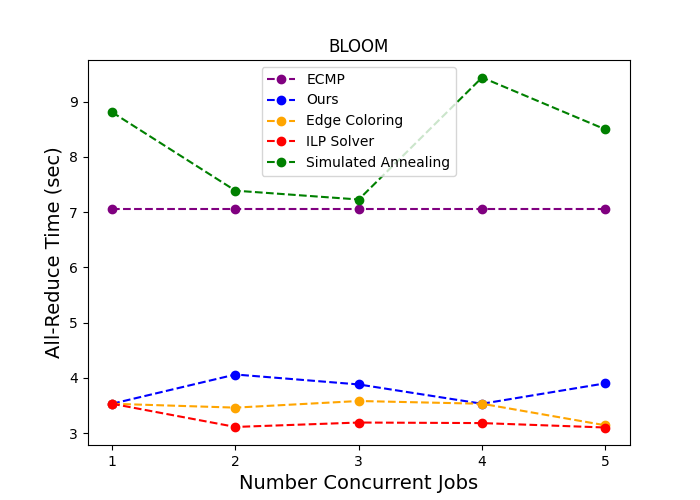}
        \caption{8 Core Failures, Ring size 4}
        \label{fig:all-reduce-time-ring-4-core-failures-8}
    \end{subfigure}
    \hfill
    \begin{subfigure}{0.49\textwidth}
        \includegraphics[width=1.0\textwidth]{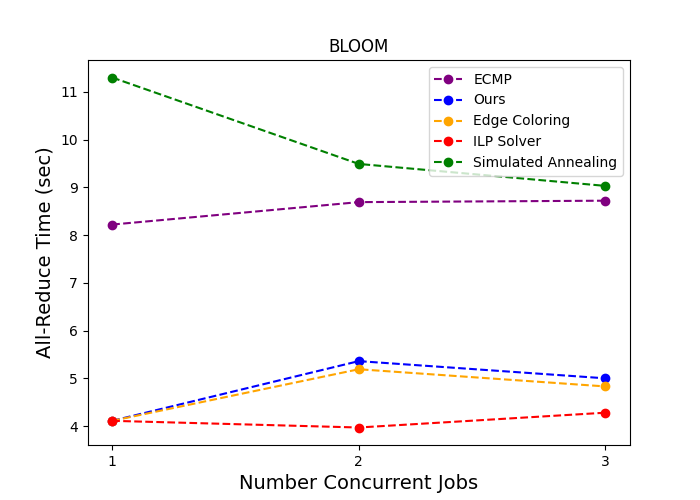}
        \caption{8 Core Failures, Ring size 8}
        \label{fig:all-reduce-time-ring-8-core-failures-8}
    \end{subfigure}
    \caption{All-Reduce (Parameter Synchronization) time of Bloom~\cite{workshop2023bloom} under various Core Failures.}
    \label{fig:all-reduce-time-core-failures}
\end{figure*}

\subsection{Runtimes}

Figure \ref{fig:computation-time-measurement} presents a comparison of runtimes, in terms of the computation latency when a job stops communicating or starts communicating. The number of commodities to be assigned to a paths varies from $100$ flows to over $1500$. Not surprisingly, our simple algorithm significantly outperform the computation heavy alternatives. In particular, the computation time for our algorithm is roughly $10$ms across the different numbers of commodities. 

Recall that our scheme incorporates forward-looking optimizations, where the controller optimizes path assignments \emph{before} a new flow arrives/leaves. In such scenarios, the online overhead of path assignments is negligible: the sum of the network RTT (O(microseconds)), the time it takes the controller to issue the \textit{pre-computed} path assignment, and the time it takes the hosts to enforce it. Our results for runtimes show that even if the controller is ``surprised'' by the traffic conditions (say, due to the arrival of a new job), the online runtime is low ($±10$ms). As discussed above, for training jobs such as those targeted here, the duration of data delivery for a single training iteration can last seconds, rendering the online runtime of our algorithm negligible in comparison. In fact, our results show that even for smaller training jobs, where data delivery for training iterations only requires $100$s of milliseconds to complete, our approach can be of value.

We also emphasize that, as explained above, the computation time for our scheme can be further accelerated by using stronger/dedicated hardware and by parallelizing computation.

\subsection{Performance Under Failures}
Figure \ref{fig:all-reduce-time-core-failures} shows how our routing scheme performs, in terms of All-Reduce time, when spine switches (randomly chosen) fail. When such a failure is detected, the controller re-runs the computation (if need be)  and sends the updated path assignment to the hosts. Once again, our scheme comes very close to the optimum (with much lower runtimes).
\section{Our Approach \emph{vs.} TopoOpt}

\begin{figure*}
    \centering
    \includegraphics[width=\textwidth]{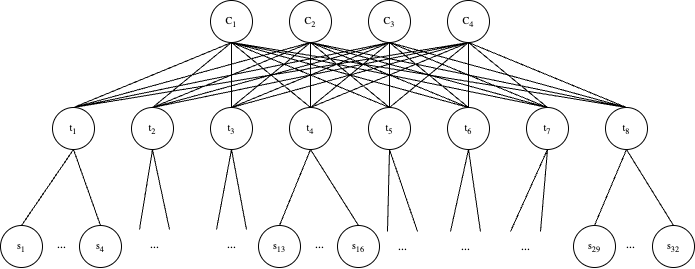}
    \caption{Simple 2-Layer Clos network where TopoOpt could fail}
    \label{fig:comparison-topoopt}
\end{figure*}

TopoOpt~\cite{wang2022topoopt} proposes utilizing optical switching to optimize the topology with respect to a parallelization strategy. This enables communicating hosts/GPUs to be directly connected. In TopoOpt, hosts/GPUs are directly connected to optical switches, providing the flexibility to interconnect them in arbitrary ways. We discuss below two major limitations of TopoOpt and how these are addressed in our framework. 

\vspace{0.1in}\noindent{\bf Slow reconfigurability.} Whether TopoOpt is implemented using patch panels or optical circuit switching (OCS), the configuration time needed to change the network topology can be slow. Indeed, as noted in~\cite{wang2022topoopt}, ``OCSs can potentially be used to reconfigure the topology of a job within training iterations, whereas patch panels are only suitable when the topology remains intact throughout the entire training of a particular job. Our evaluations demonstrate that the reconfiguration latency of today’s OCSs is too high
for today’s DNNs, leading to sub-optimal performance when the topology is reconfigured within iterations. As a result, given that faster technologies are not yet available, TOPOOPT uses a one-shot reconfiguration technique... TOPOOPT then reconfigures the interconnection between training servers of each
job before the job starts and keeps the topology intact until
the training is complete (or to recover from failures).'' TopoOpt is thus incapable of taking advantage of freed bandwidth resulting from the termination of a specific training \emph{iteration} or to reallocating network resources when an \emph{existing} job starts a new training iteration. TopoOpt is also slow to respond to failures for these reasons.

\vspace{0.1in}\noindent{\bf Limited expressiveness.} TopoOpt concentrates on scenarios with hundreds of hosts/GPUs. To scale up, TopoOpt proposes using ToRs as an intermediate layer between the hosts/GPUs and the optical switches. As explained below, this, once again, gives rise to the problems addressed by our framework. Consider the simple $2$-layer Clos network in Figure \ref{fig:comparison-topoopt}. Suppose that there are two training jobs, namely, $j_1$, and $j_2$, each trained in a hybrid parallelization manner. Suppose that these jobs induce commodity $(s_1,s_{13})$ for job $j_1$ and commodity $(s_2,s_{14})$ for job $j_2$. If TopoOpt attempts to directly interconnect the relevant ToRs, $t_1$ and $t_4$, say, through the optical switch $c_1$, without utilizing additional paths, the two commodities will collide. To achieve optimal All-Reduce time, the two commodities must use link-disjoin paths. This implies that solutions like ours must be leveraged.

\section{Related Work}

To facilitate more communication-efficient DNN training, prior research has investigated how to reduce the size of the DNN parameters transmitted across the network~\cite{alistarh2017qsgd,8868152,chung2017accelerating,goyal2018accurate,hashemi2018tictac,iandola2016firecaffe,jayarajan2019prioritybased,jia2018highly,mudigere2023softwarehardware,harlappipedream,wang2019blink}, as well as parallelization strategies that take into account the available network bandwidth~\cite{addanki2019placeto,alistarh2017qsgd,jia2018data,harlappipedream,tarnawski2020efficient}. 
While these proposals co-optimize computation and communication, they do not consider the \emph{physical network topology} as an optimization dimension. In addition, these proposals consider a single training job, ignoring the scenario that multiple jobs co-exist. 

Recent work, namely TopoOpt~\cite{wang2022topoopt}, proposes a direct-connect DNN training system that co-optimizes network topology and parallelization strategy. Topology adaptation in TopoOpt relies on reconfigurable optical switches and patch panels. We point out that, however, that TopoOpt has two significant drawbacks: \textbf{(i)} Realizing TopoOpt involves significant operational and monetary overheads and, in particular, in-network hardware changes (transitioning to optical switches), and \textbf{(ii)} TopoOpt supports ToR-to-ToR direct connectivity. When different hosts/GPUs under the same ToR must reach destinations under different ToRs, The benefits of TopoOpt are limited.

Prior studies, starting with Hedera~\cite{10.5555/1855711.1855730}, proposed dynamic flow scheduling systems for multi-stage switch topologies found in data centers. These studies differ from ours in terms of the global optimization objective (derived from our ML-training-specific motivation), the algorithms used, the implementation (ours is host-side only), and more.

\section{Discussion}

\noindent{\bf Augmenting our scheme with host-based rate control.} An interesting approach for enhancing the role of the host in our scheme is to enforce rate limiting at hosts. Specifically, when the controller computes a path-assignment, it can also compute, for each flow, the ``bandwidth share'' allocated to that host. This share could then be enforced by hosts through the transport-layer, e.g., by setting a maximum congestion window. Specifically, if the bandwidth share is set to be $x$ (in Mbps), the congestion window $cwnd$ can be set such that $\frac{cwnd}{RTT}=(1+\epsilon)x$, where $\epsilon$ is some small ``slack'' value.

\noindent{\bf Predicting flow starts and terminations, and how this impact online runtimes.} As discussed above, the controller can perform forward-looking optimizations based on predictions of when the current training iterations will terminate and when new training iterations will begin. By doing so, the controller can avoid the online runtimes associated with optimizing after the fact. We emphasize, however, that even in the absence of accurate predictions, our scheme can still make timely decisions (requiring $±10$ms to compute a new path assignment, which should be contrasted with the duration of a training iteration, which can be in the order of $100$s on ms or more). We leave the thorough investigation of prediction mechanisms in this context to future research.

\noindent{\bf Accelerating online runtimes.} We showed that our scheme quickly outputs a path assignment (in a matter of milliseconds). We believe that further acceleration can be achieved by parallelizing our scheme and running it on stronger hardware. We leave this for future exploration. (Again, we stress that when predictions of flow start/termination are accurate, online runtimes, with the exception of network and configuration delays, can be avoided altogether.)

\noindent{\bf Extending $2$-layer max-min fairness to more layers.} Our global notion of optimality reflects the desiderata of maximizing the minimum flow within a job (the first layer) and maximizing fairness across jobs (the second layer). However, tensor and pipeline parallelism give rise to another layer, as explained next. Suppose that the model is divided into two phases (sets of consecutive layers), and each of these is replicated across hosts/GPUs to support data parallelism. For the model parameters to synchronize, both phases must be synchronized. Hence, the synchronization time is dictated by the slowest phase to finish synchronizing. In our framework this translates to another (lower) layer of max-min fairness. In fact, our evaluation already reflects this $3$-layered optimization objective.

\noindent{\bf Replacing the path-assignment algorithm with an ML model (and safety considerations).} Our solution involves applying a simple and robust greedy algorithm to assign paths to commodities. We chose this algorithm since it comes with provable guarantees and is fast to execute online (and can be made faster through stronger/dedicated hardware and parallelizing the computation). One approach for accelerating the computation even further would be replacing the greedy algorithm with an ML model (e.g., a DNN) and exploiting the fast inference time of such models. The ML model could either be trained to optimize path assignments directly (potentially even outperforming our greedy algorithm in terms of solution quality), or learn to mimic our greedy algorithm (imitation learning). Employing deep learning to this end, however, would require putting in place mechanisms for ensuring the safety of this ML component, e.g., offline formal verification. (See companion Huawei project for investigation of such techniques and machinery.)

\section{Discussion}

We have presented a system for adapting routing in ML training clusters in a manner that optimizes a global notion of training efficiency and fairness. We evaluated our approach both theoretically and empirically, evidencing its usefulness and promise. While our focus has been on very large ML models, we believe that our approach is applicable also to ML models of more modest sizes, and also to models ran on all-purpose datacenter fabrics. Identifying the model sizes for which our approach is effective, and extending our ideas to the scenario that ML traffic must compete with ECMP-controlled cross traffic and to Clos networks (fat tree topologies) with more than $2$ layers.

\bibliographystyle{plain}
\bibliography{\jobname}

\end{document}